\documentclass{article}
\usepackage{times}
\usepackage{graphicx}
\usepackage{latexsym}

\usepackage{float}
\usepackage{amsthm}
\usepackage{amssymb, amsfonts, amsmath, url}
\usepackage{algorithm}
\usepackage{subcaption}

\usepackage{algorithmic}
\newtheorem{theorem}{Theorem}
\newtheorem{definition}{Definition}

\begin{document}

\title{Social Network Analysis Using Coordination Games}

\author{Radhika Arava\footnote{Nanyang Technological University}}

\maketitle

\begin{abstract}
Communities typically capture homophily as people of the same community share many common features. This paper is motivated by the problem of community detection in social networks, as it can help improve our understanding of the network topology and the spread of information. Given the selfish nature of humans to align with like-minded people, we employ game theoretic models and algorithms to detect communities in this paper. Specifically, we employ coordination games to represent interactions between individuals in a social network. We represent the problem of community detection as a graph coordination game. We provide a novel and scalable two phased approach to compute an accurate overlapping community structure in the given network. We evaluate our algorithm against the best existing methods for community detection and show that our algorithm improves significantly on benchmark networks (real and synthetic) with respect to standard \textit{normalised mutual information} measure. 
\end{abstract}

\section{Introduction}

In a social network, community structure is important as it helps in understanding the spread of information, rumor, fashion, or a video in that network. Future popularity of a meme can be predicted by quantifying its early spreading pattern in terms of its community concentration. For example, the more communities a meme permeates, the more viral it is. Thus, information about community structure gives predictive knowledge about what information will spread widely \cite{weng2013virality}. Understanding the community structure of metabolic or biological networks can help in identifying the cure for complex diseases \cite{guimera_functional_2005}. In large protein-protein interaction data sets, protein-protein interactions enable proteins to act concertedly to carry out their functions. Community structure on protein interaction networks should show potentially good candidates for the modules of proteins responsible for such functions.

In real-world networks, each vertex is likely to be a part of more than one community. For example, in social networks like Facebook, Google+, there can be overlap in user's high-school friends' circle and his university friends' circle. Similarly, in a collaboration network, a researcher can be a part of two different collaboration groups due to varied research interests and due to his different affiliations. In a biological network on genes, a gene can be responsible for more than one disease. So, it is common to observe overlap in the community structures of several real world networks. 

Most of the existing overlapping community detection algorithms fall into the following categories: (a) inference based approach \cite{gopalan2013efficient} (b) label propagation \cite{gregory_finding_2010}, \cite{xie2011slpa} (c) information theory \cite{esquivel2011compression} (d) game theory based algorithms \cite{chen2010game}, \cite{bei2013trial} (e) clique percolation methods \cite{palla_uncovering_2005}. The clique percolation algorithms like \cite{palla_uncovering_2005} are not scalable to large graphs.  Gopalan {\em et al} \cite{gopalan2013efficient} has shown that the random walk approach \cite{esquivel2011compression} did not detect the overlapping community structure accurately compared to the other algorithms considered in this paper.  The inference based approach  \cite{gopalan2013efficient} and label propagation based methods \cite{gregory_finding_2010}, \cite{xie2011slpa} are scalable and are quite fast for large networks, but are not as efficient in detecting overlapping communities as our algorithm, as shown in this paper on several different networks. Hoefer et al \cite{hoefer2007cost} proposed max-agree games, as a kind of clustering games which are similar to the graph coordination games considered in this paper. Max-agree games differ from our games such that it benefits the players to play different strategies whenever they are not connected.

There are a few game theory based overlapping community detection algorithms in the literature, and they are not scalable for large networks. The algorithm by Chen {\em et al} \cite{chen2010game} uses game theory and tries to optimize a personalized modularity function in every step using local search and its worst case running time is $O(m^2)$ ($m$ is the number of edges) time for unweighted graphs. It has been shown by Xie {\em et al} \cite{xie_overlapping_2013} that this game theory based algorithm in \cite{chen2010game} does not detect overlapping community structures accurately compared to the other existing algorithms. Bei {\em et al} \cite{bei2013trial} uses a trial and error approach along with coordination games from game theory to detect disjoint communities in a given network and is not scalable and not efficient compared to our algorithm. Similarly, Narayanam {\em et al}  \cite{narayanam2012game} used game theory to detect only disjoint communities in a network and the algorithm is also not scalable. Though our algorithm solves coordination games to detect the overlapping community structure, it is scalable to large networks and is run on 500,000 vertex networks for varying parameters. It is also found to be efficient over all the above stated algorithms on several networks for varying parameters.


In this paper, we provide a novel, scalable two-phase algorithm to compute the overlapping community structure of a network. We evaluate the algorithm \textit{NashOverlap} against the current state of the art and we find that our algorithm works far better than the best existing methods on the standard LFR benchmark networks, generated by varying all the parameters. Our algorithm is also run on small and large real world graphs like karate, dolphin and computer science bibliography database. Normalized Mutual Information is used to evaluate the quality of detected community structures against the ground-truth. Our algorithm is also evaluated against the current state of the art of the disjoint community detection algorithms and we found its performance to be as good as the other algorithms. It is the first game theory based overlapping community detection algorithm which does not optimize the modularity function and can still detect an accurate community structure for a given social network. Though our algorithm solves coordination games, it is scalable to large networks unlike the other game theoretic based algorithms.

\section{Background}
\textbf{Weighted potential games} were first introduced by Monderer and Shapley in their landmark paper~\cite{monderer1996potential}.  A game with a finite number of strategic players say $n$ is a weighted potential game if it admits a weighted potential function. 

Let $V$ be the set of all players. Let $S_i$ be the set of strategies of player $i$ and the utility function of player $i$ is $u_i:S\mapsto \mathbb{R}$ where $S=S_1 \times S_2 \times \ldots \times S_n$. $S_{-i}$ is the set of all possible strategies of all players other than player $i$. In a strategy profile $s=(s_i, s_{-i})$, $s_i$ is a strategy of player $i$ and $s_{-i}$ constitutes a strategy profile of all the players other than player $i$. Let $w=(w_i)_{i\in V}$ be a vector of positive real numbers called weights.  A function $\Phi:S\mapsto \mathbb{R} $ is a weighted potential function if for every $i \in V$, and for every $s_{-i} \in S_{-i}$, and for every $s_i, s'_i \in S_i$,
\begin{equation}
 \Phi(s'_i, s_{-i}) - \Phi(s_i, s_{-i}) = w_i \cdot (u_i(s'_i, s_{-i}) - u_i(s_i, s_{-i}))
\end{equation}

A strategy profile is called a (pure) {\em Nash Equilibrium} (local optima) if no player $i$ can improve his utility by unilaterally changing his own strategy (i.e., adopting another strategy). Every weighted potential game has a pure Nash equilibrium \cite{monderer1996potential}.\\

\noindent \textbf{Coordination games} are a class of games in which players have to choose the same strategies in order to maximize their utility. 
Typically, a coordination game can have multiple strategy assignments to all the vertices which are pure strategy Nash equilibria. 

A strategy-assignment is Pareto-efficient, if there exists no other strategy assignment in which a player can improve its utility without making other player's utility worse. Pareto-efficient assignment is a Nash equilibrium which is favored by all the players, referred to as focal point by Thomas Schelling. However, every coordination game is not guaranteed to have a pareto-efficient strategy assignment.\\

\noindent A \textbf{local search problem} is defined by  
\begin{itemize}
\item a set of problem instances $\Pi$
\item for every problem instance $I \in \Pi$
\begin{itemize}
\item a set $F(I)$ of feasible solutions
\item a value function (also called as potential function) $\Phi : F(I) \mapsto \mathbb{R}$ that maps every feasible solution $S \in F(I)$ to some value $\Phi(S)$.
\item for every feasible solution $S \in F(I)$, a neighborhood $N(S,I) \subseteq F(I)$ of $S$.
\end{itemize}
\end{itemize}

Neighborhood $N(S,I)$ of a problem $I$ is defined as all feasible solutions obtained when only one player is allowed to change its strategy. To solve a local search problem is to find a feasible solution $S \subset F(I)$ that is a local optimum, i.e., $\Phi(S) \leq \Phi(S')$ for every $S' \in N(S,I)$. 

The problem of solving a coordination game (which is also a potential game) is equivalent to the problem of solving an equivalent local search problem~\cite{kleinberg2006algorithm}. It means, the set of Nash equilibria of coordination game is the same as the set of local optima for the local search optimization problem.

\section{Community Detection Problem}

In a social network $G=(V,E,w)$, $V$ is the set of vertices and $E$ is the set of undirected edges that represents relationships between vertices in the network and $w:E \mapsto \mathbb{R} $ is the weight function on edges. If the graph is unweighted, we assume that each edge has unit weight. A community in a social network is a non-empty connected subset of vertices with denser connections within themselves than with the rest of the network. Formally, if $C_i \subseteq V$ denotes a community $i$, then our goal is to identify the community structure $\Gamma = \{C_1,C_2,\ldots, C_m\}$, where $\bigcup_i C_i = V$.  

\subsection{Community Detection as a Graph Coordination game}

As indicated earlier, people have selfish motives and prefer to align with like-minded people and hence choose the same strategies as that of their closely bonded friends. According to Schelling segregation model \cite{schelling1971dynamic}, these micro-motives of people to bond with people of similar interests and choose strategies accordingly leads to a macro-behavior emerging in the network, in the form of a community structure. Therefore, to account for selfish motives of people to align with the like-minded, we represent the problem of identifying community structure as a game where the set of actions available to a player is the set of all communities. 

We consider a game for each set of connected players. Each player has the set of all communities in the network as its set of strategies. The utility for each player, when they play the same strategy (choose the same community) as that of his neighbor in the coordination game, is proportional to their tie-strength \cite{granovetter1973} (formally defined later) and the utility when they play different strategies is zero. Extending on this utility definition, the utility of a player playing the same strategy as that of a subset of his friends (neighbors in graph) is proportional to the sum of his tie-strengths with each of them. This is a \textbf{graph coordination game}.


\section{Approach}

We design a novel and scalable approximation approach in two phases to compute a community structure of the network by solving for a particular Nash equilibrium of the formulated graph coordination game. 

In the first phase, we solve $k$ graph coordination games independently by formulating the local search version $\zeta_1$ of each game in the following way. Each vertex has a set $S$ of $r$ strategies. Assign a strategy picked uniformly from set $S$ of $r$ strategies for each vertex $v \in V$. Pick a uniform random ordering of vertices and each vertex gets its turn to maximize its utility according to this ordering. In its turn, the vertex $i$ picks the strategy $s_i \in S$, whichever gives the maximum utility. This is repeated with the same ordering of vertices until no vertex can increase its utility by changing its strategy.

We prove that this graph coordination game $\zeta_1$ is a potential game. Any Nash equilibrium of a $r$-strategy coordination game results in at most $r$ communities. The local optimum of $\zeta_1$ of any two independent games depends on the vertex-ordering chosen and the initialization of the strategies to the vertices. For uniform initialization of strategies to vertices and any vertex ordering in $\zeta_1$, we have already seen that when the network has a clear community structure, its local optimum detects the community structure with a high probability.

To deal with the networks whose community structure is fuzzy, we compute the proportion of $k$ games in which an edge chooses the same strategy to estimate their probability of choosing the same community.

From these $k$ games, we compute the edge-closeness value for each edge (i.e., probability that an edge can be in the same community) and an intermediate partition of the network formed by considering only edges with a reasonably high edge-closeness value. So, we consider the edges with edge-closeness value $>$ 0.95, as the edges which are most probably the community-edges (edges with both its vertices in a community). Using these edges, we construct the set of all the connected components and call it as intermediate partition. So, the connected components in the intermediate partition correspond to the communities in which all the pairs follow the same strategy, in almost all the games.

In the intermediate partition, the cut-edges are the edges where the vertices are in different communities and community-edges are the edges where both vertices are in the same community. Since, we identify an edge as community-edge only if it survives as a community-edge in more than $95\%$ of the graph coordination games in the first phase, an actual cut-edge is more likely to be identified as a cut-edge in the intermediate partition. However, an actual community-edge can be mistakenly identified as a cut-edge, depending upon the fuzziness of the network. So, using the intermediate partition and edge-closeness values computed in the first phase, the second phase takes care of such mistakenly identified edges and outputs an accurate and stable overlapping community structure.

We show that the problem of computing the local optimum for all these games is NP-Hard using the ideas from \cite{kleinberg2006algorithm}. However, we can allow for minor changes in the game parameters and show that we can compute a stable overlapping community structure in linear time (linear in number of edges).

The overall algorithm, referred to as $\textit{NashOverlap}$ is best illustrated in two phases as shown in the Figure \ref{fig:algo}. 

\begin{figure}[ht]
\includegraphics[width=\textwidth]{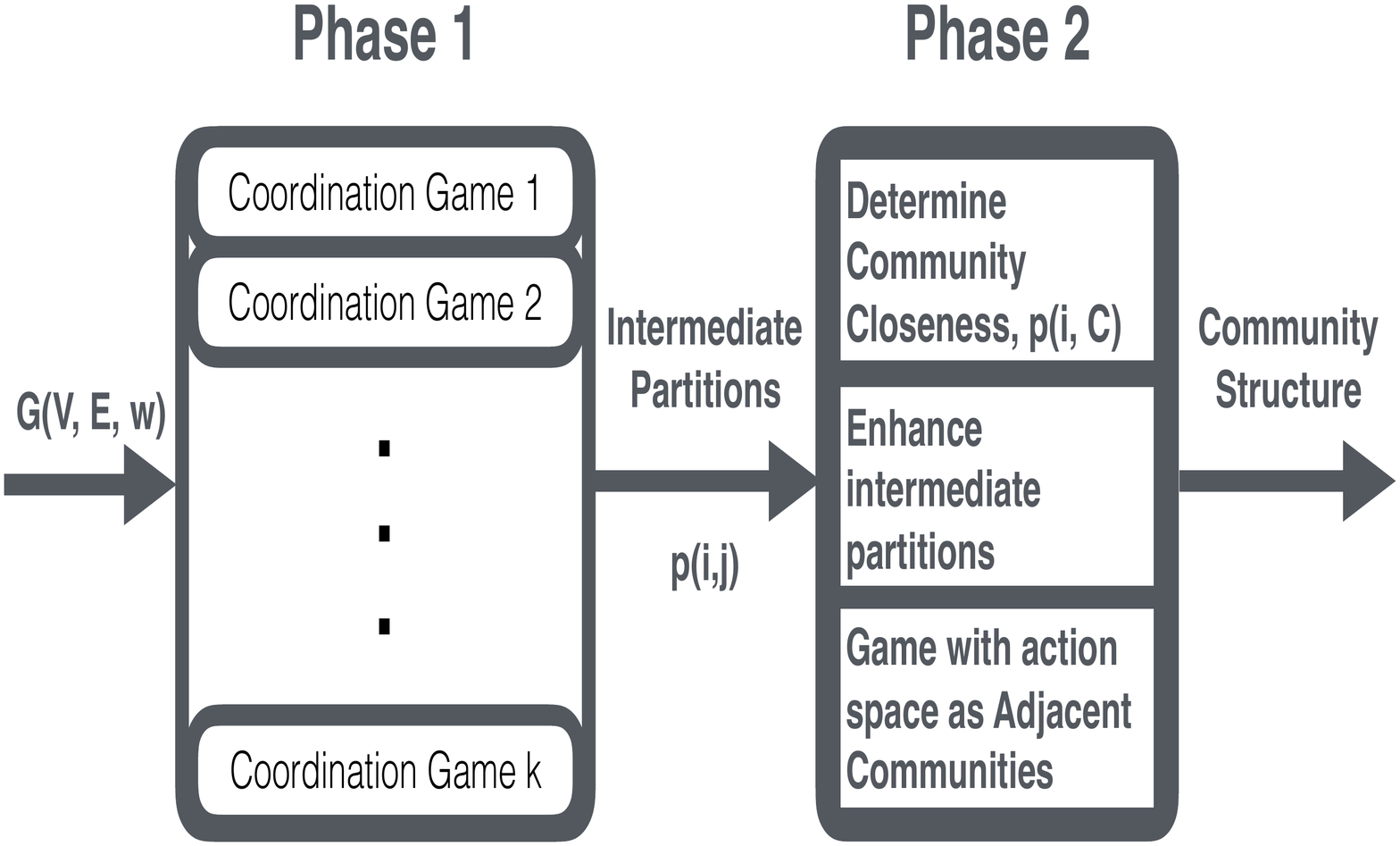}
\caption{Algorithm: NashOverlap}
\label{fig:algo}
\end{figure}

\subsection{First Phase}

In this phase, we identify edge-closeness values and an intermediate partition of the graph using the concept of tie-strength. The concept of tie-strength was first introduced by Granovetter in his landmark paper \cite{granovetter1973} and their hypothesis is that stronger the tie-strength between two people, the larger their number of common friends and their weights. Using that hypothesis, we define the tie-strength of any pair of friends in our model which measures how tightly the pair is bonded to each other. Specifically, tie-strength $t(i,j)$ is the sum of the weights due to all common friends $k$ of $i$ and $j$ plus the edge weight $w(i,j)$. Given a weighted network $G(V, E, w)$, tie-strength $t(i,j)$ of an edge $(i,j)$ is defined as:
\begin{equation}
t(i,j) = w(i,j) + \displaystyle \sum_{\substack{{k:(i,k) \in E}\\{(k,j) \in E}}} \left( w(i,k) + w(j,k) \right).
\end{equation}

We can then define $t_i$, which is the sum of the weights of $i's$ adjacent edges:
\begin{equation}
t_i = \displaystyle \sum_{j:(i,j)\in E}\;t(i,j)
\end{equation}

As $t(i,j)$ measures the strength of bonding between $i$ and $j$, we assume that utility for both $i$ and $j$ due to each other is proportional to their tie-strength when they choose the same strategy and their utility is 0 when they choose different strategies. This serves as the utility matrix of a coordination game.  Formally, let $s_i \in \{0,1,\ldots,r-1\}$ denote the strategy that player $i$ chooses; then $(s_i)_{i \in V}$ defines a strategy profile of the game. Let $s_{-i}$ denote the strategies of all the other players other than $i$. Utility of a player $i$ at a given strategy profile $s=(s_i, s_{-i})$ is given by
\begin{equation}
 u_i(s_i, s_{-i}) \propto \displaystyle \sum_{\substack{{j :(i,j)\in E}\\{s_i = s_j}}} t(i,j)
\label{eq:utility}
\end{equation}

We have chosen $\frac{1}{t_i}$ as the proportionality constant for our experiments, assuming that utility of a player depends only on his and his friends' strategies. This forms a strategic game  $\zeta^1 = (V, (S_i)_{i \in V}, (u_i)_{i \in V})$, where $V$ is the set of players, $S$ constitutes the set of $r$ ($r\geq 2$)strategies for each player $i$ and $u_i$ is the utility function for each player $i$ as defined in the equation (\ref{eq:utility}). We now show that $\zeta^1$ is a potential game with the potential function $\Phi_1:\lbrace 0, 1,\ldots,r \rbrace ^{|V|}  \mapsto \mathbb{R}$ that is defined as:
$$\Phi_1 = \displaystyle \sum_{\substack{(i,j) \in E \setminus E_c}} t(i,j),$$
where $E_c$ is the set of cut-edges.

\begin{theorem}
$\Phi_1$ is a weighted potential function.
\end{theorem}

To ensure, we obtain edge closeness values, we solve $k$ independent local search problems, each with a different initialization. $k$ is chosen to be 100 for all our experiments and depends on the size of the set of strategies for each player. The more the strategies for all the players, the lesser the value of $k$. 


\begin{definition}
Edge-closeness value $p(i, j)$ for edge $(i, j)$ is the proportion of $k$ games in which every pair of friends $(i, j)$ have chosen the same strategy. 
\end{definition}


These $p(i,j)$ values are representative of the macro-behavior emerged due to the micro-motives of each player to align with people having similar interests. Using the edge-closeness values for all the edges, we compute an intermediate partition, which is the set of connected components in the network considering only the edges whose edge-closeness values is at least $\beta$ ($\beta$ is set to $0.95$ for all experiments). $p(i,j)$ values can be less for an edge with a reasonably good tie-strength. The correlation of $p(i,j)$ and tie-strengths is discussed in the subsequent subsection.


\subsection{Second Phase}
This phase takes care of the mistakenly identified edges in the intermediate partition of first phase and decides their communities appropriately using edge-closeness values and outputs a final overlapping community structure. To measure the closeness of the communities to each vertex appropriately, we extend the notion of edge-closeness and define community-closeness of a vertex $i$ to a community $C$, $p(i,C)$ as follows:

\begin{definition}
Given a vertex $i$ and a community $C$, define the community-closeness, $p(i,C)$ to be the sum of edge-closeness values between $i$ and $C$. That is, \[p(i,C) = \sum_{\substack{{j: (i,j) \in E}\\{j\in C}}} p(i,j).\]
\end{definition}

Intuitively, a vertex is said to be adjacent to a community, if it has any of its adjacent vertices in that community. For a given overlapping community structure $\Gamma$, we shall duplicate every vertex in each of its overlapping communities. Let ${\cal C}_i$ be the set of communities of player $i$ in a given overlapping community structure. So, we put a copy of vertex $i$ in each of its communities $C \in {\cal C}_i$. Define, the utility of the player $i$ as the sum of its community-closeness values to its communities. 

Consider a game $\zeta^2 = (V, (S_i)_{i \in V}, (u_i)_{i \in V})$ where $S_i$ is the set of all $i's$ adjacent communities, $u_i$ is the sum of $i's$ community-closeness values to all its communities. We show that this game is a potential game and compute a local optima. A trivial solution to this game is where every vertex is a part of every community. However, that solution which also happens to be the optimal solution to the game does not represent the actual community structure and is not interesting. So, we compute the local optima of this game which represents a stable community structure in the following way.  We also show that this local optima represents a community structure as long as the network has less mixing factor \cite{lancichinetti2009benchmarks}. 

Pick a random vertex ordering and each vertex in its turn, computes its community-closeness to each of its adjacent communities. The vertex chooses to be a part of those communities to which its community-closeness is at least $\alpha$ times its maximum community-closeness, if its utility increases with its decision. $\alpha$ is a overlap parameter and lies in $[0,1]$ and it is decreased to increase the overlap in the detected community structure. If $\alpha=1$, the resulting community structure has no overlap. For the purpose of solving the problem of disjoint community detection, we can set $\alpha=1$. For our experiments on synthetic networks, the results at $\alpha=0.5$ are almost same as the best results obtained at any $\alpha \in [0.3,0.7]$. 

The local optimum of this algorithm gives a stable overlapping community structure. This is because, for a network with a clear community structure, we get an intermediate partition which is a community structure without overlap. To deal with the overlapping vertices, we assume that the vertices which are overlapping have their internal degree equally distributed among its communities. We choose the communities for each vertex by computing their maximum community-closeness and allotting the communities to which its community-closeness is at least $\alpha=0.5$ times its maximum community-closeness. If a vertex is not overlapping, we suppose that all other communities of that vertex are less than $\alpha$ times it maximum community-closeness for a reasonably high $\alpha$. However, if two communities have many inter-community edges, that together the two communities look like a random graph, then the two communities are returned as a single community.


Using the intermediate partition and edge-closeness values from the first phase and using the above algorithm, 

$\zeta^2$ is a potential game with potential function $\Phi_2: \lbrace {\cal C}_1, {\cal C}_2, \ldots, {\cal C}_{|V|} \rbrace \mapsto \mathbb{R}$ that is given by

\begin{equation}
\Phi_2 = \frac{1}{2}\cdot \displaystyle \sum_{\substack{{i \in V}\\{C \in {\cal C}_i}}} p(i,C)
\end{equation}

\begin{theorem}
 $\Phi_2$ is a weighted potential function.
\end{theorem}

\begin{proof}
When player $k$ changes his communities from $\zeta_k$ to $\zeta'_k$, let us say that the cover changes from $\Gamma$ to $\Gamma'$.
Then the difference in the potential function is as follows :
\begin{equation}
\frac{1}{2}\cdot \left( \displaystyle \sum_{\substack{{i \in V}\\{C \in {\cal C}_i}\\{{\cal C}_i \in \Gamma'}}} p(i,C) - \displaystyle \sum_{\substack{{i \in V}\\{C \in {\cal C}_i}\\{{\cal C}_i \in \Gamma}}} p(i,C)\right)
\end{equation}
Note that $p(i,C)=0$ for all $C$, which are not adjacent communities of player $i$.  
The above expression is same as
\begin{equation}
\sum_{\substack{{(i,j) \in E}\\{C \in {\cal C}'_i \cap {\cal C}'_j}\\{(i,j) \in C}}} p(i,j) - \sum_{\substack{{(i,j) \in E}\\{C \in {\cal C}_i \cap {\cal C}_j}\\{(i,j) \in C}}} p(i,j)
\end{equation}
where ${\cal C}_i$ and ${\cal C}'_i$ denotes the set of communities of player $i$ in the cover $\Gamma$ and $\Gamma'$ respectively.
When player $k$ changes its communities, then only the $k's$ adjacent edges $(k,l)$ need to be taken into account. 
\begin{equation}
\sum_{\substack{{l :(k,l) \in E}\\{C \in {\cal C}'_k \cap {\cal C}_l}\\{(k,l) \in C}}} p(k,l) - \sum_{\substack{{l :(k,l) \in E}\\{C \in {\cal C}_k\cap {\cal C}_l}\\{(k,l) \in C}}} p(k,l)
\end{equation}
which is same as the 
\begin{equation}
\sum_{\substack{C \in {\cal C}'_k}} p(k,C) - \sum_{\substack{C \in {\cal C}_k}} p(k,C)
\end{equation}
So, the above equation gives the difference in the utility of the player $k$ when it changed its set of communities from ${\cal C}_k$ to ${\cal C}'_k$, which is always positive. So, this defines a weighted potential game and hence the game converges. 
\end{proof}

At equilibrium, we obtain a stable and accurate overlapping community structure detected by our algorithm.

\subsection{Discussion}
One trivial equilibrium of this game that is also the optimal solution of the game is one where all the players choose the same strategy. However, given the selfish nature of individual nodes, this equilibrium is highly unlikely.

We show here that when a network has a clear community structure say $\Gamma$, with $c$ number of communities, then there exists a Nash equilibrium which represents the community structure $\Gamma$ of the network for the above game. We provide an approach to solve the game for a local optimum which represents the community structure $\Gamma$ with high probability. 

Assume that in $\Gamma$, there are few inter-community edges and each community in the network is very dense such that it is assumed to be a complete graph. Each vertex has the same set $S$ of strategies and $|S| = r$. Assign a random assignment of say $r$ strategies/colors/communities independently to all the vertices. In the initial assignment of colors, each vertex gets a given color with a probability of $1/r$. Then in each of $c$ communities, the distribution of vertices over $r$ colors is uniform. So, every color is equally likely to be present in each community and distribution of the number of vertices over all colors in each community is uniform. According to a uniform random vertex-ordering, allow vertices to make decisions on their strategies sequentially to increase their utility. Allow the vertices to make decisions according to that ordering and the vertex-ordering can be repeated until no vertex can increase its utility unilaterally. 

In one-walk of vertex ordering, where vertices make decisions sequentially, consider the sub-walk of vertices of a given community. Since there are almost no inter-community edges, there is almost zero influence on the decisions of vertices in any community due to the decisions made by vertices outside the community. 

\begin{theorem}
Assume, the community is a complete graph. For any vertex ordering, given such a community, the sequential decisions made by the vertices to maximize their utility according to their ordering, always leads to the local optimum in which all the vertices in the community gets a unique color. 
\end{theorem}

\begin{proof}
Let the given community has $f$ number of vertices. It is assumed as a complete graph on $f$ vertices. We prove that all the vertices in the community gets a unique color by contradicting it. Suppose that the community has $r'$ colors at local optimum where $1 < r' \leq r$. Let blue be the color of maximum number of vertices and red be the second best (ties are broken arbitrarily if two colors have the same number of vertices). Pick a red colored vertex. Since each vertex is connected to every other vertex, the red colored vertex is connected to maximum number of blue vertices and hence should have been colored blue at local optimum. So, it contradicts that the community has more than one color in local optimum.
\end{proof}

Given a color, the probability that all the vertices in a single community gets that color is 1/r. The probability that two communities get the same color is 1/r. The probability that all the communities gets unique colors is $\frac{\binom{r}{c} \cdot c!}{r^c}$. As $r \gg c$, each community gets a unique color with high probability. If each color represents the community, the local optimum represents the community structure $\Gamma$.

Thus we can compute a community structure of the network using the above algorithm for a network with a clear community structure and dense communities which we assume as complete graph. We repeat this game for enough number of times say $k$ and compute the edge-closeness values for each edge. Edge-closeness value for each edge is defined as the proportion of games such that an edge chooses the same community at equilibrium in any game. In this particular scenario when the network has a clear community structure, the expected value of edge-closeness value is 1 for all community-edges of the network and $1/r$ for all the cut-edges over all the $k$ games. The connected components considering the edges with edge-closeness value at least 0.95 represent the actual community structure $\Gamma$.

However, in reality, the actual communities are not always enough dense and the network usually do not possess a very clear community structure. So, the actual edge-closeness values for all the edges vary a bit from their expected values.

Hence, our algorithm to detect the fuzzy and overlapping community structure is formulated in two phases. 

\subsection{Algorithm Analysis}

We compute the pearson correlation between the edge-closeness values and edge tie-strengths. We find that there is a positive correlation between edge-closeness values and edge tie-strengths in all the cases. However, the correlation is not linear and it increases as the mixing factor increases as shown in the figures (\ref{fig:scatmu10om2}) and (\ref{fig:scatmu50om2}). At lesser mixing factors, all the edge-closeness values are distributed around 0.5 and 1. It is easy to differentiate between a cut-edge and a community-edge. As the mixing factor increases, the fuzziness of the community structure increases. The community edges are harder to detect when the network is fuzzy as their edge-closeness values spread far away from 1. So, a community-edge is more likely to be identified as a cut-edge, with increase in mixing factor, the edge-closeness value need not be more than 0.95 for a community-edge when the network is fuzzy. 

The distribution of edge-closeness values for varying mixing factors is also shown in the figures (\ref{fig:histmu10om2}) and (\ref{fig:histmu50om2}). For lesser mixing factors, the number of cut-edges is very less compared to that of the number of community edges. Also, the number of edges with edge-closeness values lying between (0.55 and 0.95) is almost negligible. With increase in the mixing factor, the number of cut-edges are almost same as the number of community edges. However, the number of edges with edge-closeness value lying between (0.55, 0.95), increases.

\begin{figure}[t!]
 \begin{subfigure}[b]{0.5\textwidth}
        \centering
        \subcaption[short for lof]{Scatterplot for mu = 0.1}
        \includegraphics[height=1.8in]{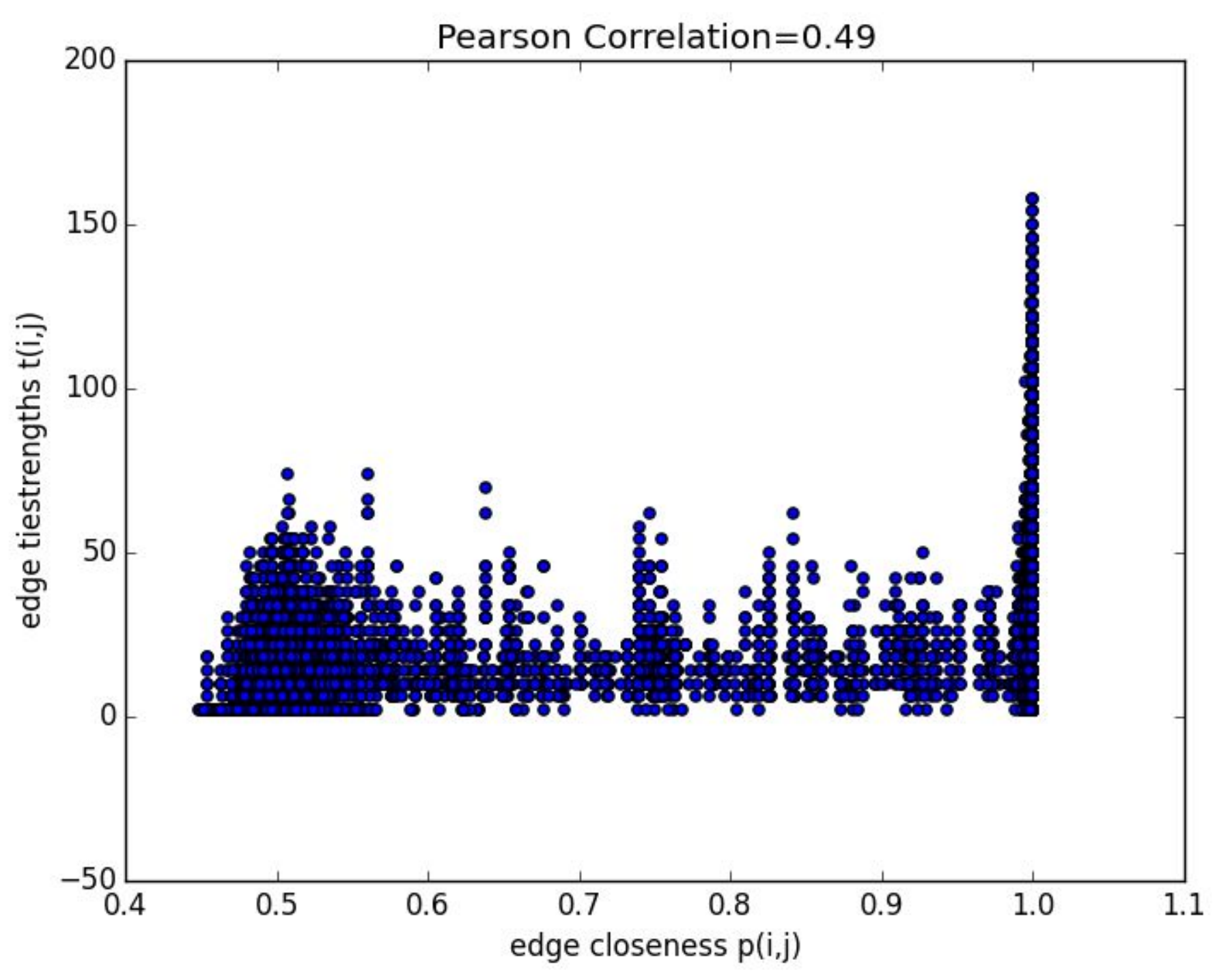}
        \label{fig:scatmu10om2}
  \end{subfigure}
    \quad
    \begin{subfigure}[b]{0.5\textwidth}
        \centering
        \subcaption[short for lof]{Histogram for mu = 0.1}
        \includegraphics[height=1.8in]{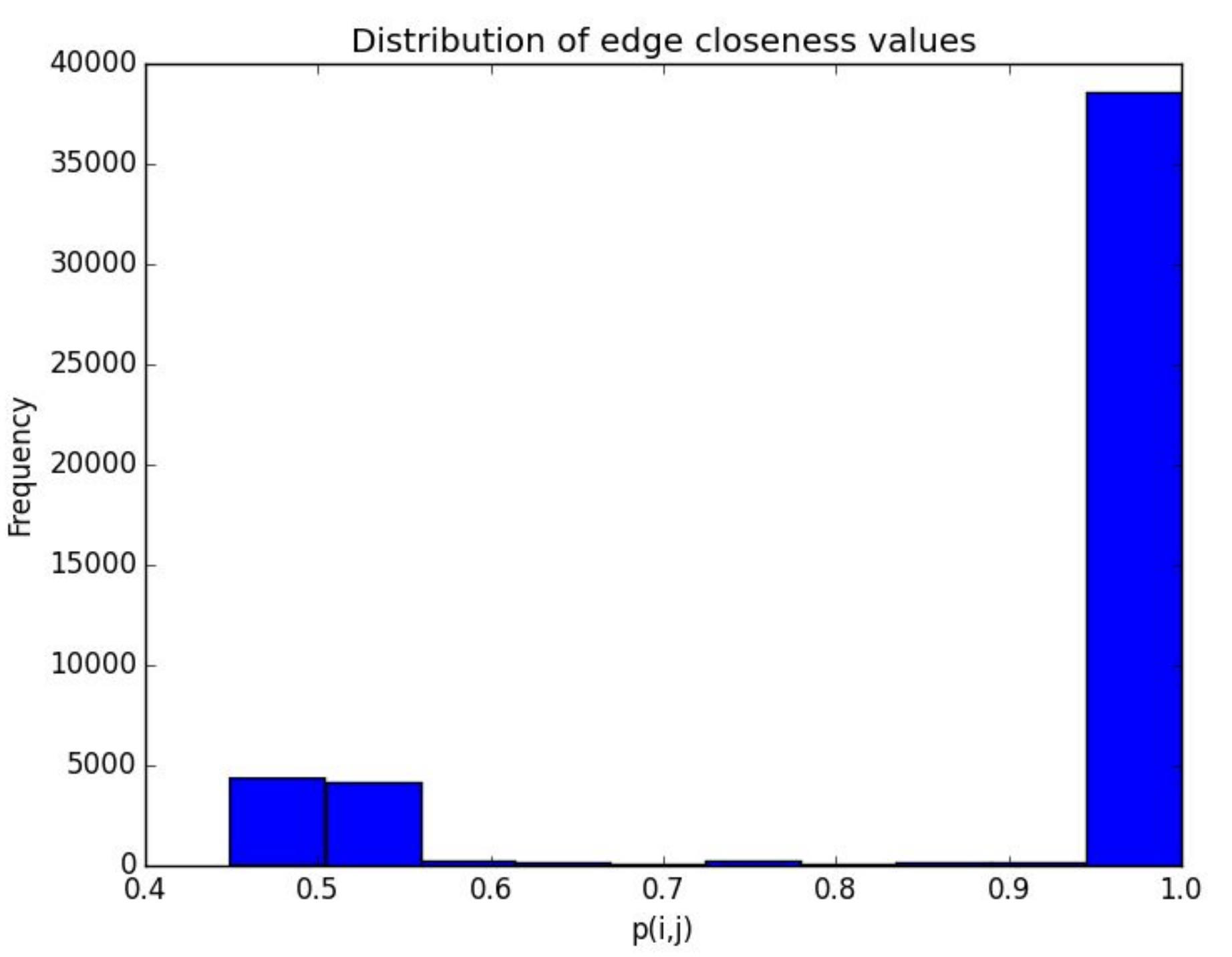}
        \label{fig:histmu10om2}
    \end{subfigure}
%
~
 \begin{subfigure}[b]{0.5\textwidth}
        \centering
         \subcaption[short for lof]{Scatterplot for mu = 0.5}
        \includegraphics[height=1.8in]{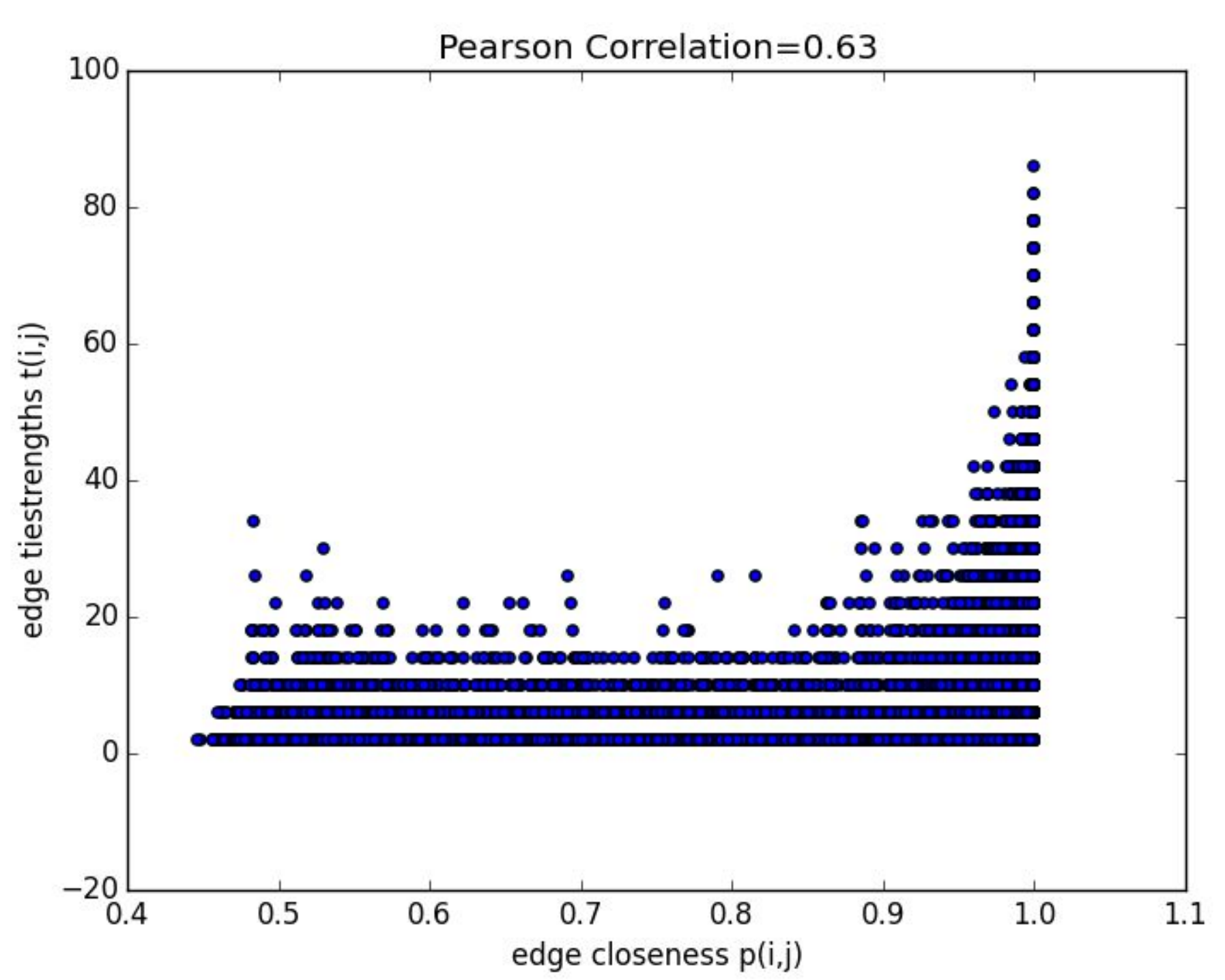}
        \label{fig:scatmu50om2}
    \end{subfigure}
    \quad
    \begin{subfigure}[b]{0.5\textwidth}
        \centering
        \subcaption[short for lof]{Histogram for mu = 0.5}
        \includegraphics[height=1.8in]{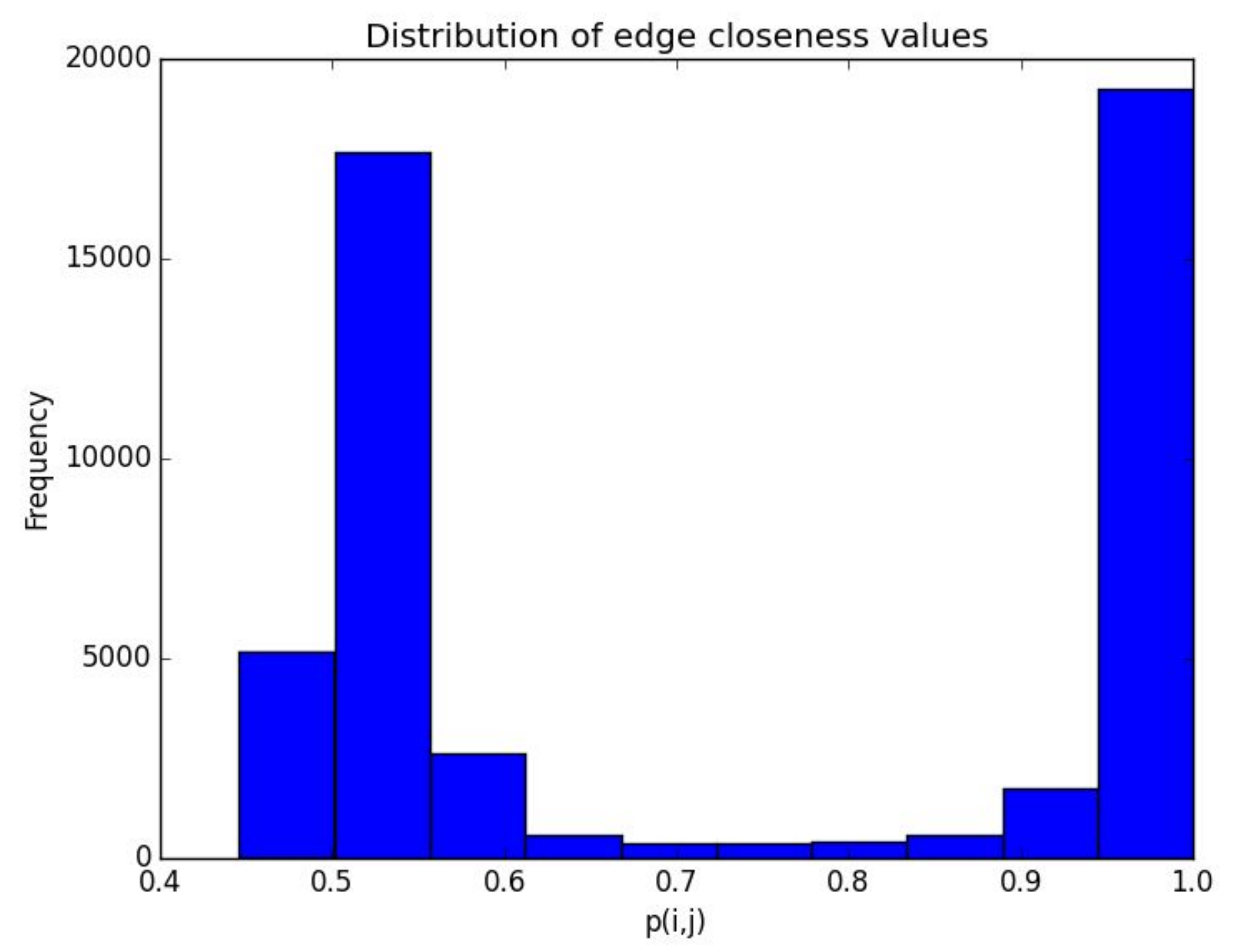}
        \label{fig:histmu50om2}
    \end{subfigure}
  \caption{Scatterplots and distribution of edge-closeness values for 5000 vertex network when $r=2$}
    \label{fig:mu50om2}
\end{figure}

\begin{figure}[t!]
 \begin{subfigure}[b]{0.5\textwidth}
        \centering
        \subcaption[short for lof]{Scatterplot for mu = 0.1}
        \includegraphics[height=1.8in]{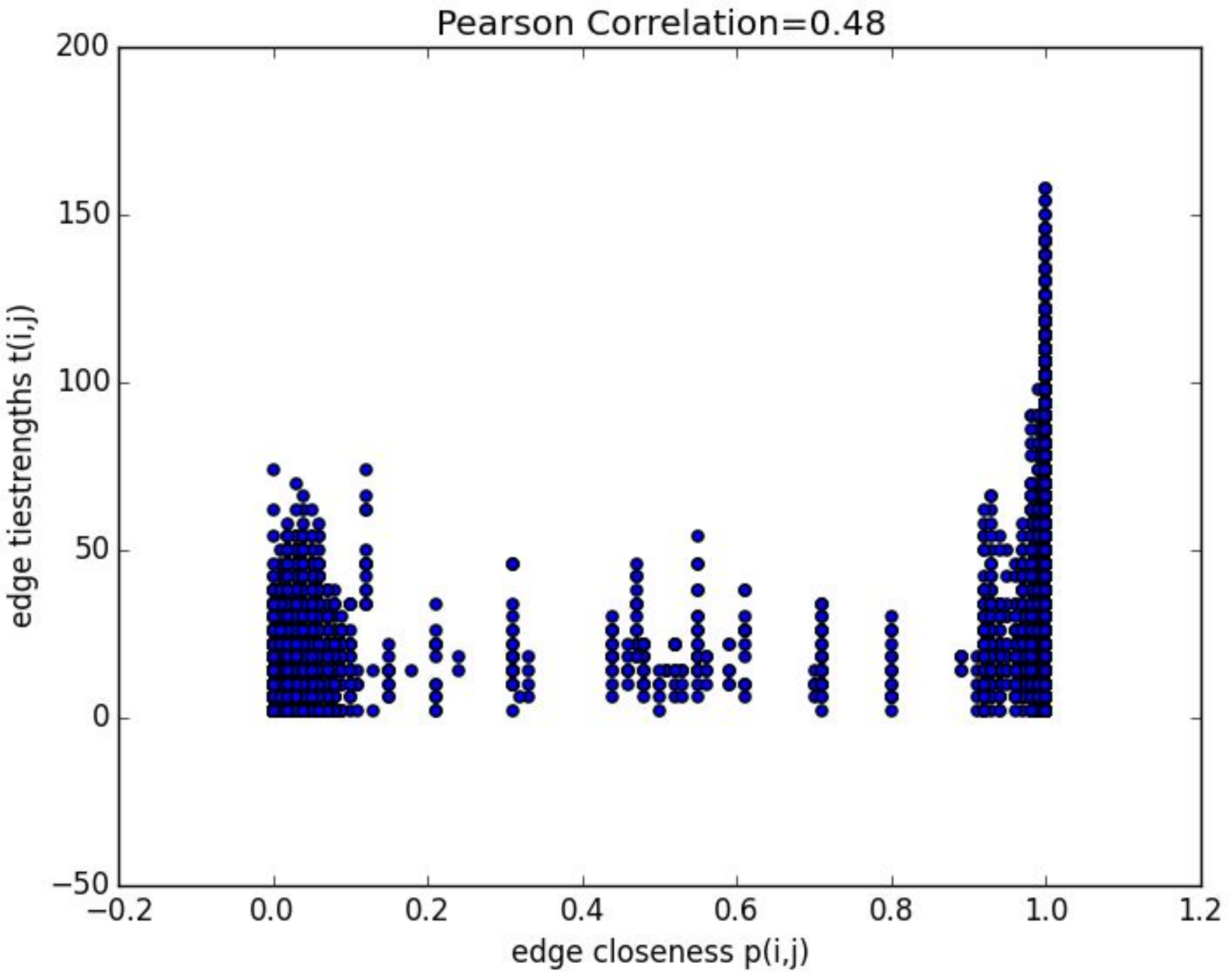}
        \label{fig:scatmu10om2r40}
  \end{subfigure}
    \quad
    \begin{subfigure}[b]{0.5\textwidth}
        \centering
        \subcaption[short for lof]{Histogram for mu = 0.1}
        \includegraphics[height=1.8in]{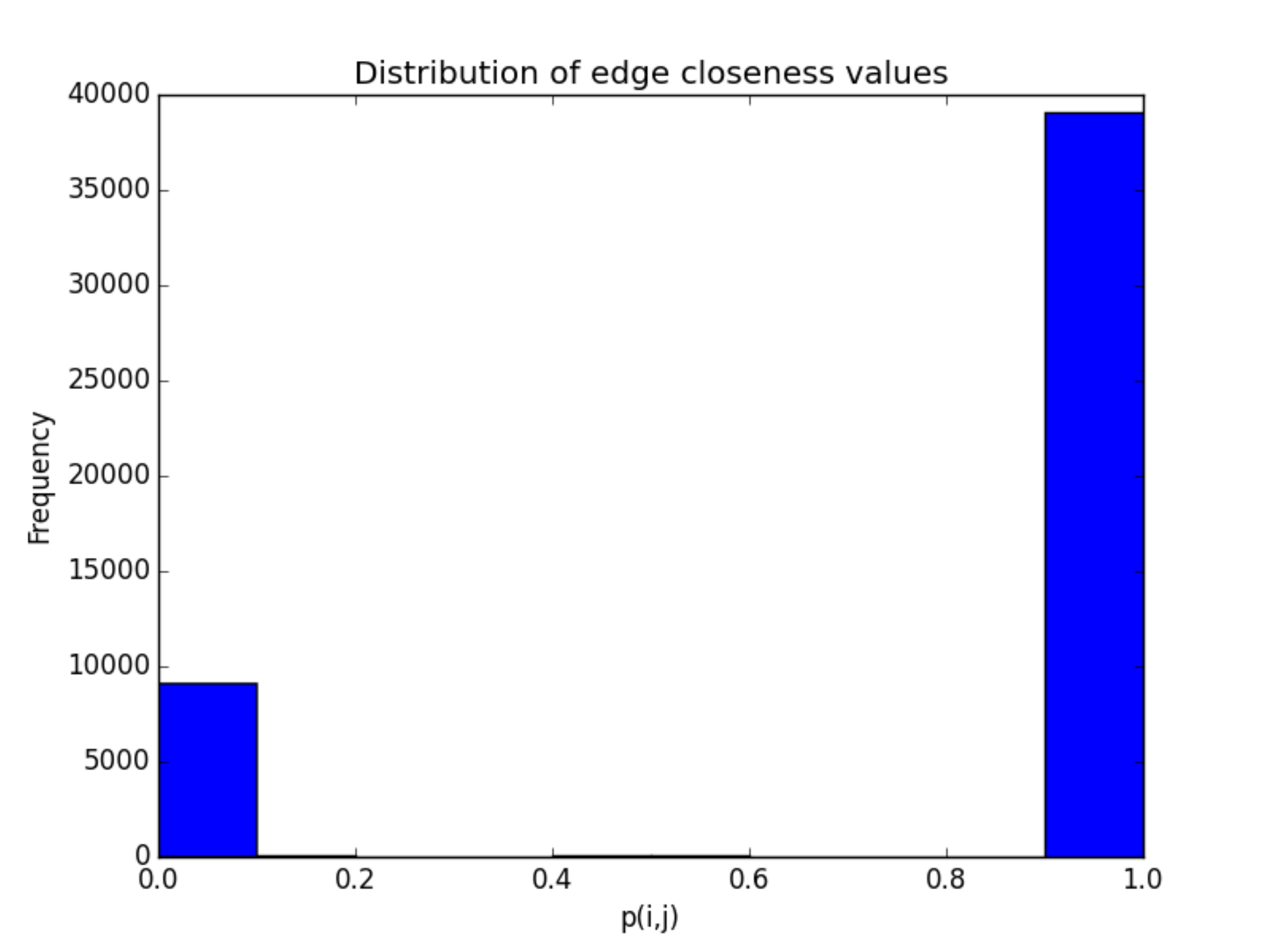}
        \label{fig:histmu10om2r40}
    \end{subfigure}
%
~
 \begin{subfigure}[b]{0.5\textwidth}
        \centering
         \subcaption[short for lof]{Scatterplot for mu = 0.5}
        \includegraphics[height=1.8in]{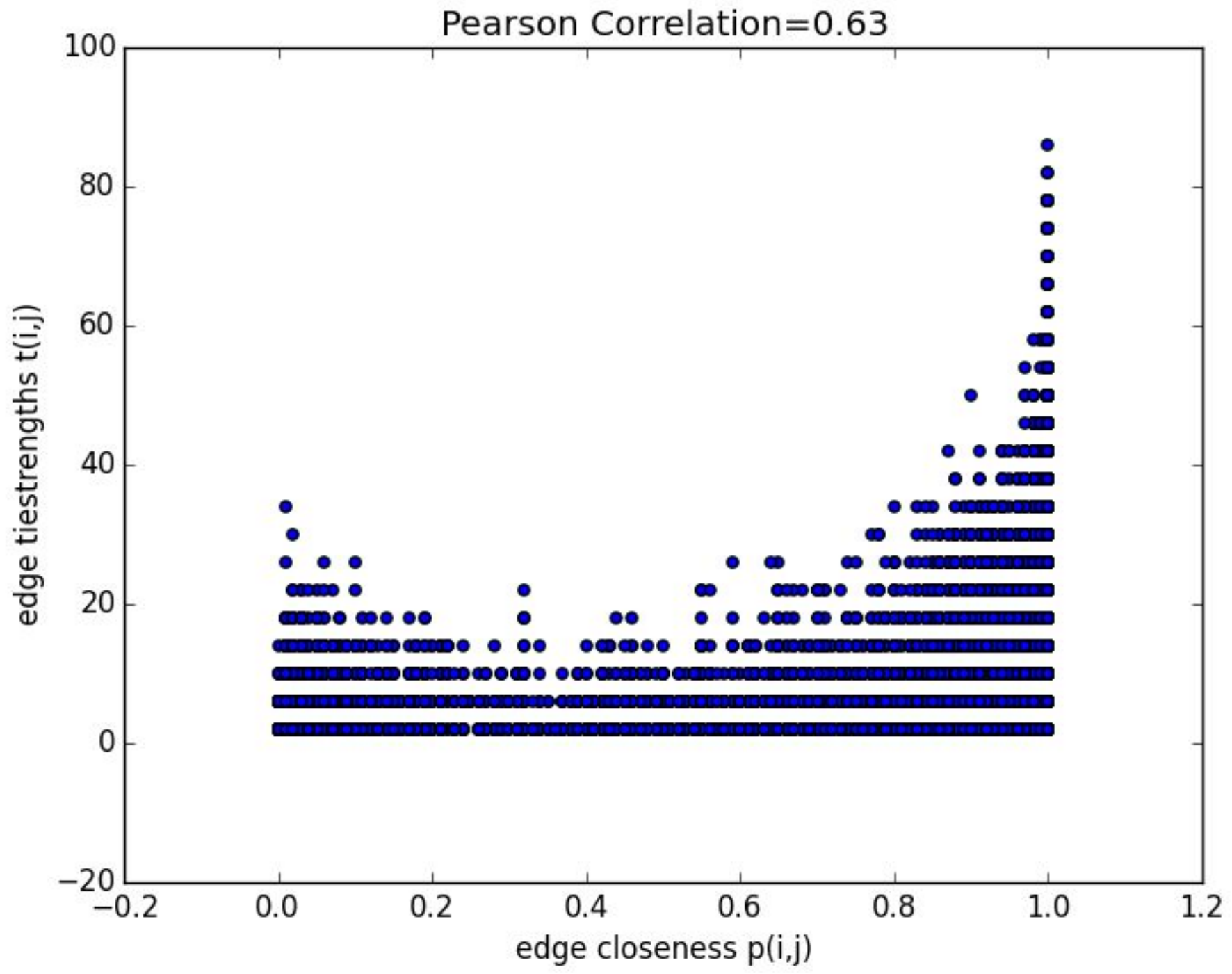}
        \label{fig:scatmu50om2r40s}
    \end{subfigure}
    \quad
    \begin{subfigure}[b]{0.5\textwidth}
        \centering
        \subcaption[short for lof]{Histogram for mu = 0.5}
        \includegraphics[height=1.8in]{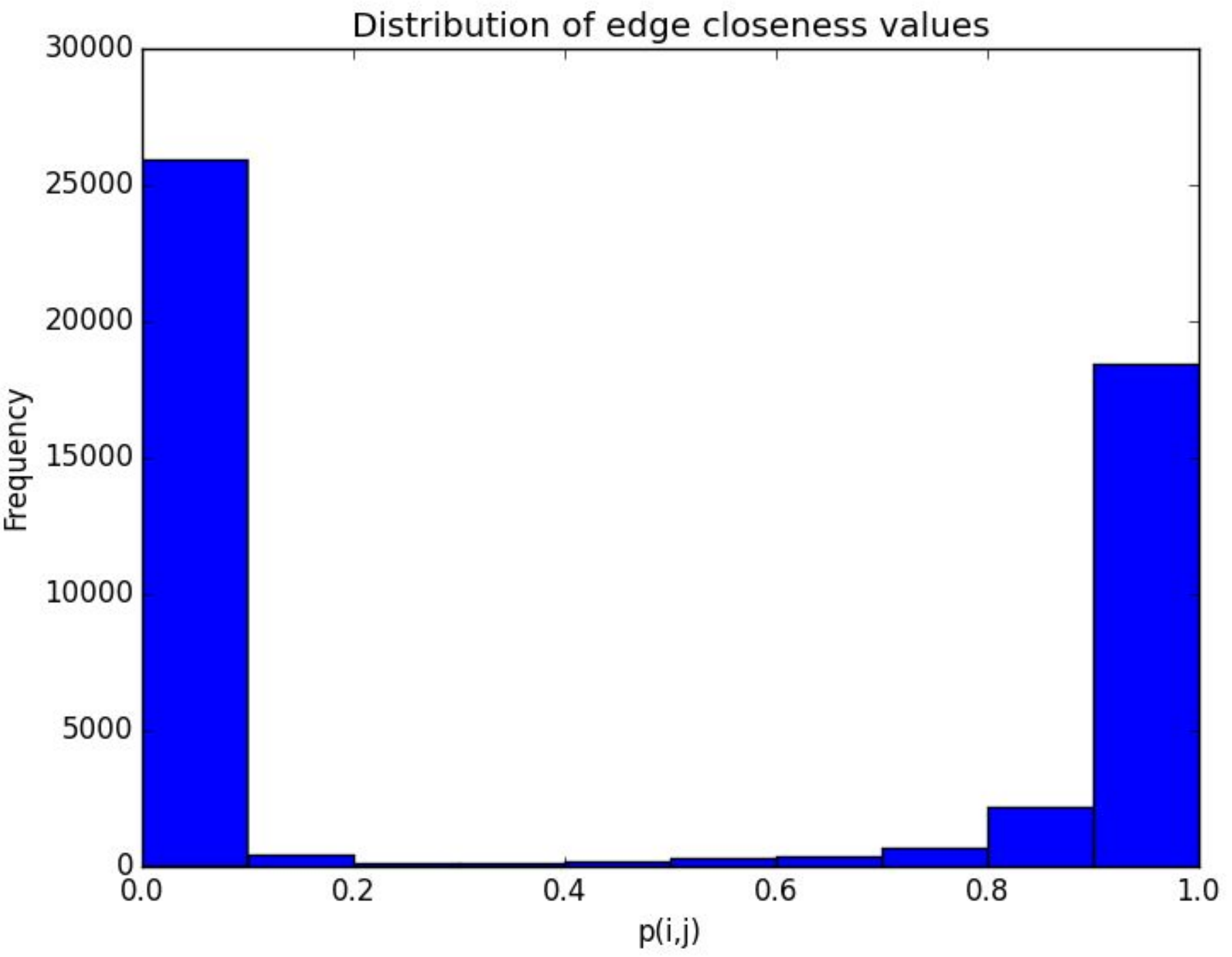}
        \label{fig:histmu50om2r40}
    \end{subfigure}
  \caption{Scatterplots and distribution of edge-closeness values for 5000 vertex network when $r=40$}
    \label{fig:mu50om2}
\end{figure}

\begin{theorem}
The problem of computing the local optimum of the graph coordination game formulated in our paper is NP-Hard. 
\end{theorem}

However this game can be run faster by allowing only big-enough improvements taking the ideas from \cite{gaur2008capacitated} and \cite{kleinberg2006algorithm}. 


\begin{definition}
An improvement in the potential function is a good-enough improvement if a vertex changes its strategy only if the potential function increases by a factor of $\frac{2 \cdot \epsilon}{n}$ for some $\epsilon > 0$ where $n$ is the number of vertices.
\end{definition}

The entire network is the globally optimal solution. So, the maximum value of the potential function in the first phase is the sum of all the edge tie-strengths. We compute the estimate of the local optimum obtained by our algorithm in the first phase using the following theorem.

\begin{theorem}
Let $s$ be a local optimum. Then if we allow only good-enough improvements while solving the game, we get $(2+\epsilon) \cdot \Phi_1(s) \geq W$
\end{theorem}
\begin{proof}
Let $u_i$ be the utility of vertex $i$ in any game in the first phase. At a local optimal strategy profile of vertices, let the partition be $(A_1,A_2, \ldots, A_r)$.  Since only good enough improvements are allowed, the following equations hold true $\forall h \in \lbrace 1, 2, \ldots, r \rbrace$,
\begin{equation}
 \sum_{i \in A_h} u_i(s) \geq \sum_{i \in A_h} \left(1-u_i(s) - 2\cdot \epsilon \cdot \Phi_1(s)/n \right)
\end{equation}

Adding the above $r$ equations, we get the following
\begin{equation}
2 \cdot \Phi_1(s) \geq 2\cdot(W-\Phi_1(s)) - 2 \cdot \epsilon \cdot \Phi_1(s)
\end{equation}
where $W=\sum_i t_i$.

We get, 
\begin{equation}
\Phi_1(s) \geq \frac{W}{2+\epsilon}
\end{equation}
\end{proof}

\begin{theorem}
The greedy algorithm which accepts good-enough improvement terminates in at most $O(\epsilon^{-1}\cdot n \cdot log(W))$ rounds, where $W=\sum_i t_i $ in the first phase and terminates in at most $O(\epsilon^{-1}\cdot n \cdot log(W))$ rounds, where $W=\sum p(i,j)$ in the second phase.
\end{theorem}

\begin{proof}
Each improvement increases the objective function by at least a factor of $(1 + \frac{\epsilon}{n})$. Since $(1 + \frac{1}{x})^x \geq 2$ for any $x \geq 1$, we see that $(1 + \frac{\epsilon}{n})^{(\frac{n}{\epsilon})} \geq 2$, and so the objective function increases by a factor of at least 2 every $\frac{n}{\epsilon}$ flips. The weight cannot exceed $W$, and hence it can only be doubled at most $log W$ times. 
\end{proof}

%

Note that, our algorithm does not work if the local search algorithm is modified such that it always detects the global optimum in any of the two phases. However if the network has no community structure, then our algorithm $NashOverlap$ always detects the global optimum and the entire network is detected as the community structure. That is the reason, our algorithm cannot detect the community structures when the mixing factor is greater than 0.5 or when the number of overlapping vertices increases to $50\%$ accurately.

\begin{algorithm}[H]
\caption{\textit{NashOverlap}}
\begin{algorithmic}
\algsetup{linenosize=\tiny}
\small
\STATE $G(V,E,w)$ is the input network
\FOR {$v \in V$}
\STATE Pick strategy $s_v$ in $\lbrace 1, 2, \ldots, r\rbrace$ uniformly at random 
\ENDFOR
\FOR{$each\;game$}
\STATE Pick an ordering assignment of vertices say $A$
\REPEAT
\FOR {$v \in A$}
\STATE Choose $s_v$ such that $u_v(s_v,s_{-v})$ is maximum
\ENDFOR
\UNTIL partition is the same in two subsequent iterations
\STATE Update the edge-closeness values $p(i,j)$ for all edges
\ENDFOR
\STATE Compute intermediate partition such that $p(i,j)>\beta$ 
\STATE $\alpha$ is a overlap parameter to be tuned.
\REPEAT
\FOR {$v \in V$}
\STATE Let $\zeta_v$ be set of v's adjacent communities
\FOR {$C \in \zeta_v$ }
\STATE Compute the community-closeness $p_v(C)$.
\ENDFOR
\STATE Compute $\max_C~\lbrace p_v(C) \rbrace$
\STATE Let $\zeta'_v=\lbrace C|p_v(C) > \alpha.\max_C (p_v(C)) \rbrace $
\STATE Assign $\zeta'_v$ as $v$'s communities if $u_v(\zeta'_v) \textgreater u_v(\zeta_v)$
\ENDFOR
\UNTIL  Community structure is same in two subsequent iterations
\RETURN Return the overlapping community structure
\end{algorithmic}
\label{alg2}
\end{algorithm}

\section{Experiments}
We first describe the set up for the experiments and then provide the results. 

\subsection{Setup}
To study the performance of \textit{NashOverlap}, we conducted extensive experiments on LFR benchmark networks \cite{lancichinetti_benchmarks_2009} whose ground truth community structure is already known. 

Given a community structure, mixing factor $\mu$ is the maximum fraction of degree of each non-overlapping vertex, outside its community. For example, if the mixing factor is 0.1, then any non-overlapping vertex in any given community in that community structure has at most $10\%$ of its degree outside its community. Overlapping membership $om_i$ is the number of communities of the vertex $i$. Let $on$ be the number of overlapping vertices in a given network and $\gamma_1$ and $\gamma_2$ be the exponents of power law distributions for degree and community size respectively. 

 \begin{table}
 \small
  \begin{center}
 \begin{tabular} {|l| l| l|  l|  l|  l|}
 \hline 
$om$ & SLPA & COPRA & CFinder & OSLOM & Nash \\ \hline  
2 & 0.97943&0.9985&0.87659&0.99469&0.999807\\ \hline
3 & 0.93156&0.99305&0.8495&0.97704&0.998647\\ \hline
4 & 0.88207&0.96666&0.84058&0.94059&0.995646\\ \hline
5 & 0.83685&0.90756&0.82456&0.90608&0.974573\\ \hline
6 & 0.80283&0.83821&0.79269&0.86573&0.946112\\ \hline
7 & 0.7672&0.77883&0.75958&0.82451&0.90833\\ \hline
8 & 0.72346&0.71594&0.73811&0.77809&0.86996\\ \hline
 \end{tabular}
 \end{center}
 
 \begin{center}
 \begin{tabular} {|l| l| l|  l|  l|  l|}
 \hline
$om$ & SLPA & COPRA & CFinder & OSLOM & Nash\\ \hline
2 & 0.96908&0.99386&0.573&0.99145&0.98645 \\ \hline
3 & 0.92113&0.96677&0.59233&0.95551&0.982377 \\ \hline
4 & 0.86531&0.91272&0.6062&0.90197&0.960399 \\ \hline
5 & 0.81523&0.83743&0.58491&0.8488&0.914377 \\ \hline
6 & 0.76729&0.77201&0.6287&0.80385&0.861334 \\ \hline
7 & 0.72642&0.70658&0.6152&0.75626&0.808351\\ \hline
8 & 0.67675&0.65799&0.59361&0.71198&0.769981 \\ \hline
 \end{tabular}
\end{center}

 \begin{center}
 \begin{tabular} {|l| l| l|  l|  l|  l|}
 \hline
 $om$ & SLPA & COPRA & CFinder & OSLOM & Nash \\ \hline
2 & 0.87444&0.95993&0.36046&0.97164&0.968639 \\ \hline
3 & 0.79596&0.89308&0.3135&0.89679&0.94078 \\ \hline
4 & 0.75394&0.80646&0.37866&0.82748&0.877891 \\ \hline
5 & 0.692&0.74409&0.36277&0.76214&0.787768 \\ \hline
6 & 0.65122&0.69011&0.32862&0.7129&0.727489 \\ \hline
7 & 0.60913&0.64374&0.34251&0.66702&0.678984\\ \hline
8 & 0.56703&0.59485&0.34953&0.62428&0.63227\\ \hline
\end{tabular}
\caption{Comparison of algorithms on 5000 vertex networks varying overlapping membership from 2 to 8 and community sizes in the range [20, 100] for $\mu = 0.1$ and $\mu = 0.3$ and $\mu = 0.5$ respectively and $on = 10\%$.}
\label{tab:5kl}
\end{center}
 \end{table}
  
Two overlapping community structures are evaluated for their similarity using the standard measure, \textit{Normalized Mutual Information}(NMI) \cite{lancichinetti_detecting_2009}. The networks are generated by considering the following parameter values:
\begin{itemize}
\small
\item $\mu$ : Mixing factor varied between 0.1 to 0.5 in steps of 0.1
\item $n$ : Network size in $\lbrace 1000, 5000, 10000, 100000, 500000 \rbrace$
\item $om$ : Overlapping membership in $\lbrace 2, 3, \ldots, 7, 8 \rbrace$
\item $on$ : Overlapping vertices $10\%$ of the network size ($n$)
\item $maxk$ : Maximum degree is 50 ($\sqrt n$ for large $n$)
\item $k$ : Average degree is 20  ($15 \cdot om$ for large $n$)
\item $\gamma_1$ : Exponent of power law distribution for degree is 2
\item $\gamma_2$ : Exponent of power law community size distribution is 1
\item $minc$, $maxc$ : Community sizes are picked from [20, 50], or [20, 100] (for small $n$)
\item $minc$, $maxc$ : Community sizes in $[200 \cdot n/100, 500 \cdot n/100]$ (for large network sizes)
\end{itemize}

We compared the performance of our algorithm against current best algorithms: CFinder \cite{palla_uncovering_2005}, oslom \cite{lancichinetti_finding_2011}, copra \cite{gregory_finding_2010}, slpa \cite{xie_slpa:_2011}, svinet \cite{gopalan2013efficient}, Infomap \cite{esquivel2011compression}. The recent algorithms svinet \cite{gopalan2013efficient}, Infomap \cite{esquivel2011compression} did to perform very well on all the synthetic networks considered and hence are not included in our comparison results. All the results are averaged over 10 different LFR benchmark graphs for each case. We evaluate the performance of algorithm \textit{NashOverlap} with that of other algorithms with respect to a standard measure \textit{Normalized Mutual Information} \cite{lancichinetti_detecting_2009}. 

CFinder algorithm \cite{palla_uncovering_2005} is run for each of the values $k = 3, 4, 5, 6, 7, 8, 9, 10$ and the partition that returns the maximum NMI is picked. Svinet \cite{gopalan2013efficient}, Infomap \cite{esquivel2011compression} algorithms are run with parameters set to default values. The results of svinet and Infomap are not shown in the paper as the results are comparably bad to the algorithms considered in this paper. The paper by \cite{xie_overlapping_2013} has shown that the game-theoretic algorithm in \cite{chen2010game} do not detect the overlapping community structure accurately compared to the other algorithms considered in this paper. We ran the comparison against only those algorithms whose code is freely available in their respective author's websites. We ran SLPA algorithm \cite{xie2011slpa} on each value of $r$ = 0.05, 0.10, 0.15, 0.20, 0.25, 0.30, 0.35, 0.40, 0.45 for 10 times, which is a parameter which decides the number of times the algorithm can be repeated to improve its accuracy. We ran the algorithm COPRA, for each $v \in \lbrace 1,\ldots,10 \rbrace$, enabling the $mo$ for 10 times, as suggested by its manual. OSLOM is run with the parameter $r$ set to 10 and the rest of the parameters are set to default values. For the execution of \textit{NashOverlap}, we change the parameter $\alpha$ from $0.30$ to $0.70$, in steps of $0.02$ to detect the best overlapping community structure, for a given network. We picked only the best result for $\alpha \in [0.3, 0.7]$, just like we did for the rest of the algorithms. However, in our algorithm, NMI values do not differ by more than 0.1 when $\alpha=0.5$. 10 replications are performed for each initial condition. Our code is freely available at the following site: \textit{https://github.com/radsine00/cd}.


 \begin{table}
 \small
 \begin{center}
 \begin{tabular} {|l| l| l|  l|  l|  l|}
 \hline 
$om$ & SLPA & COPRA & CFinder & OSLOM & Nash \\ \hline  
 2 & 0.94081&0.99046&0.98891&0.99709&0.988725  \\ \hline 
3 & 0.88409&0.97186&0.95648&0.97162&0.984701 \\ \hline 
4 & 0.83081&0.92413&0.90241&0.93037&0.963998 \\ \hline 
5 & 0.7814&0.85657&0.86218&0.8866&0.939689 \\ \hline 
6 & 0.72565&0.7807&0.82019&0.8421&0.905594 \\ \hline 
7 & 0.67505&0.69129&0.78472&0.79773&0.849367 \\ \hline 
8 & 0.624&0.64776&0.74429&0.75376&0.843787 \\ \hline 
 \end{tabular}
 \end{center}

 \begin{center}
 \begin{tabular} {|l| l| l|  l|  l|  l|}
 \hline 
$om$ & SLPA & COPRA & CFinder & OSLOM & Nash \\ \hline  
2 & 0.94974&0.993&0.99807&0.99927&0.993333 \\ \hline
3 & 0.90199&0.97662&0.99597&0.99347&0.998454\\ \hline
4 & 0.85202&0.93949&0.97434&0.96904&0.992285\\ \hline
5 & 0.81337&0.88667&0.93541&0.93427&0.969991\\ \hline
6 & 0.76999&0.81585&0.90114&0.89681&0.936879\\ \hline
7 & 0.72725&0.74869&0.86722&0.85517&0.933189\\ \hline
8 & 0.68238&0.67474&0.84138&0.81597&0.898582\\ \hline
 \end{tabular}
 \end{center}
 \begin{center}
 \begin{tabular} {|l| l| l|  l|  l|  l|}
 \hline  
 $om$ & SLPA & COPRA & CFinder & OSLOM & Nash \\ \hline  
2 & 0.87739&0.96534&0.74218&0.97932&0.96987  \\ \hline
3 & 0.81609&0.91718&0.70824&0.92384&0.95785 \\ \hline
4 & 0.74911&0.83426&0.70452&0.86404&0.898204 \\ \hline
5 & 0.68484&0.73797&0.70125&0.80304&0.828384 \\ \hline
6 & 0.63765&0.67824&0.67734&0.75269&0.773777 \\ \hline
7 & 0.58215&0.62764&0.65471&0.70534&0.731957 \\ \hline
8 & 0.52835&0.5773&0.60169&0.65807&0.698016 \\ \hline
  \end{tabular}
 \caption{Comparison of algorithms on 5000 vertex networks varying overlapping membership from 2 to 8 and community sizes in the range [20, 50] for $\mu = 0.1$ and $\mu = 0.3$ and $\mu = 0.5$ respectively and $on = 10\%$.}
  \label{tab:5ks}
 \end{center}
  \end{table} 

\subsection{Results}
Due to space constraints, we are unable to show all the experimental results. We provide results for a few settings that are representative of the comparison results observed over all the parameter settings. As the community size increases in comparison to the network size, the fuzziness increases. So, we chose to fix the network size and change the community sizes and show that our algorithm can still detect the community structure accurately.

\textit{NashOverlap} outperformed all the other algorithms with respect to the Normalised Mutual Information (NMI) measure in most experimental settings. The improvement provided increases with increase in the network size for a given community size. While the performance of all the algorithms degrades with increase in overlapping membership, we show that  \textit{NashOverlap} has greater values for NMI compared to the other considered algorithms as shown in the Tables \ref{tab:5kl} and , \ref{tab:5ks}.

The overlapping membership also increases with the decrease in overlap parameter $\alpha$. If $\alpha$ is one, then the algorithm \textit{NashOverlap} detects a disjoint community structure, even when the network has an overlapping community structure. 

We run the algorithm with overlap parameter $\alpha = 1$ on benchmark networks \cite{lancichinetti_benchmark_2008} to evaluate disjoint community detection algorithms. We varied mixing factor from 0.1 to 0.5, and community sizes and its performance is found to be as good as the following disjoint community detection algorithms, CFinder \cite{palla_uncovering_2005}, Louvain \cite{blondel_fast_2008}, Infomap \cite{rosvall_maps_2008}, Infomod \cite{rosvall_information-theoretic_2007} for all the considered networks.

\subsection{Running times}
The number of rounds taken by any game for any network size in the first or the second phase to converge to the nash equilibrium is approximately $n\cdot k_1$ where $k_1$ is a constant and $k_1 \leq 10$ in all the considered networks. In each round, each vertex takes $O(d)$ time to compute its utility where $d$ is its degree and choose the best strategy. Roughly, the total time to compute the overlapping community structure using our algorithm is $\max \lbrace O(k \cdot m), O(n \cdot d_{max} \cdot k_1)\rbrace$, where $k$ is the number of games in the first phase and $d_{max}$ is the maximum degree of a vertex.

On a laptop with Intel core i7 processor of 2.2GHz, the algorithm $NashOverlap$ is run on various network sizes. Table \ref{tab:runningtime} summarizes the running times of the algorithm.

 \begin{table}
 \small
 \begin{center}
 \begin{tabular} {|l| l| l|  l|  l|  l|}
 \hline 
 
network size&1000 & 5000 & 10000 & 100000  & 500000\\ \hline  
running time&1&4&26&404&2402 \\ \hline 
 \end{tabular}
 \end{center}
 \caption{Comparison of the Running times (in seconds) of $NashOverlap$ on benchmark networks for small communities, and $on=10\%$ and fixing average, maximum degree at 20 and 50 respectively and $\mu=0.1$ and $om=4$}
  \label{tab:runningtime}
  \end{table}
  
\subsection{Real World Graphs}
The algorithm is also run on real world graphs like karate(n=34), dolphin(n=115), football(n=62), celegan-neural(n=297), jazz(n=198), email(n=1133), and netscience(n=1589). The community structures detected by our algorithm have comparable modularity \cite{newman_modularity_2006} values to that of those detected by Louvain \cite{blondel_fast_2008} algorithm, one of the best modularity optimization algorithms till date.

The algorithm is also run on large real world networks like DBLP, computer science bibliography database. We considered only the papers from 1966 to 2014 from around 69 premier conferences and formed a collaboration network of ~ 120000 authors and ~ 400000 papers. Weight of the edge between any two researchers in this collaboration network is the number of papers that the two researchers have co-authored together. 

We attempted to detect the collaboration groups in this weighted network using our algorithm. The detected groups are manually checked for their accuracy and are compared against the collaboration groups detected by $COPRA$ and $OSLOM$. 

Our algorithm detected a large collaboration group of size $38\%$ of the network size. The large collaboration group detected for each dataset has the researchers, mostly based in United States, from almost all the fields, and the second largest collaboration group is far smaller in size than the first. The diameter of the largest collaboration group for the large dataset is 9 and the average path length which is the average distance between any pair of researchers in the giant component for the large dataset is 5. Note that, this is because the average number of intermediate researchers that connect any pair of researchers in the giant component is quite less. This finding is in congruence with the statement from the paper \cite{newman2001structure} which says that in all the collaboration networks studied, they identified a giant component which is a considerable fraction of the network size and the second group happened to be far smaller than the first. The existence of a giant component and the subsequent smaller components is due to the percolating regime characteristic of collaboration networks. When the network has a fewer edges, there exist a large number of small connected clusters. With the increase in the number of edges, most of these small clusters group to form a giant component, leaving behind smaller clusters which join the giant component in the future with the increase in the number of connections. This leads to a reasoning that the scientific collaboration network is highly connected and there is a possibility of more interdisciplinary work in the future which is a good sign for the development of science. 

We observe that the communities formed using our algorithm are representative of the research communities (identified based on the conferences attended). Each community detected represents the set of authors with major share of their papers published only in specialized set of conferences. Roughly, in each community, around $90\%$ authors have specialization in a single research area and most of the authors in each community belong to same country/continent. We have detected several communities each of which represents the authors from data mining, multi-media, database systems, theory of computing, cadence and automation, complexity, cryptography, machine learning, logic, artificial intelligence, computer vision, and algorithms and data structures. 

Some of the collaboration groups identified are as shown in the figures \ref{fig:compgeom} and \ref{fig:dbmlai}.
\begin{figure}[ht]
\includegraphics[width=0.8\textwidth]{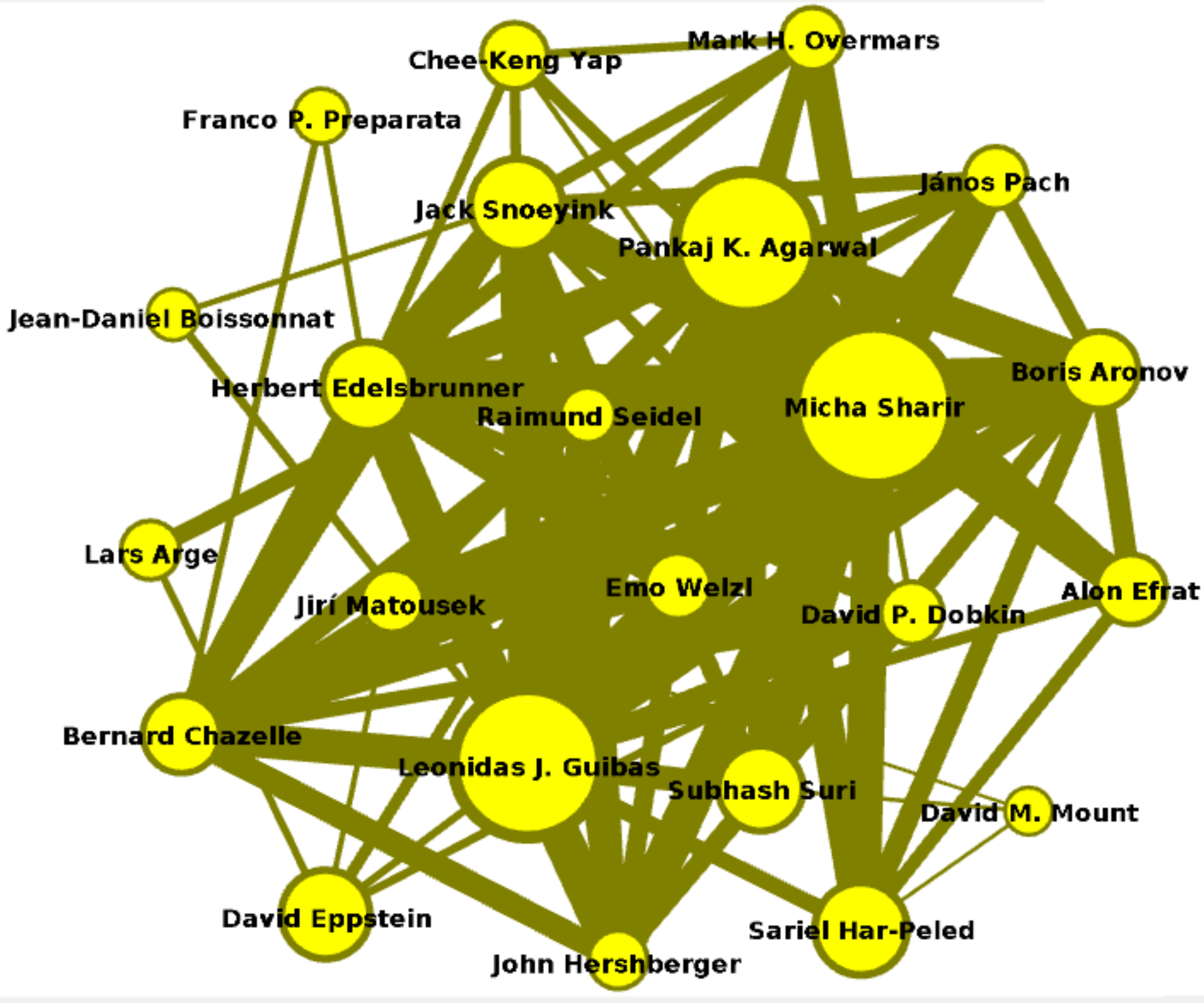}
\caption{Computational geometry collaboration group}
\label{fig:compgeom}
\end{figure}
\begin{figure}[ht]
\includegraphics[width=0.8\textwidth]{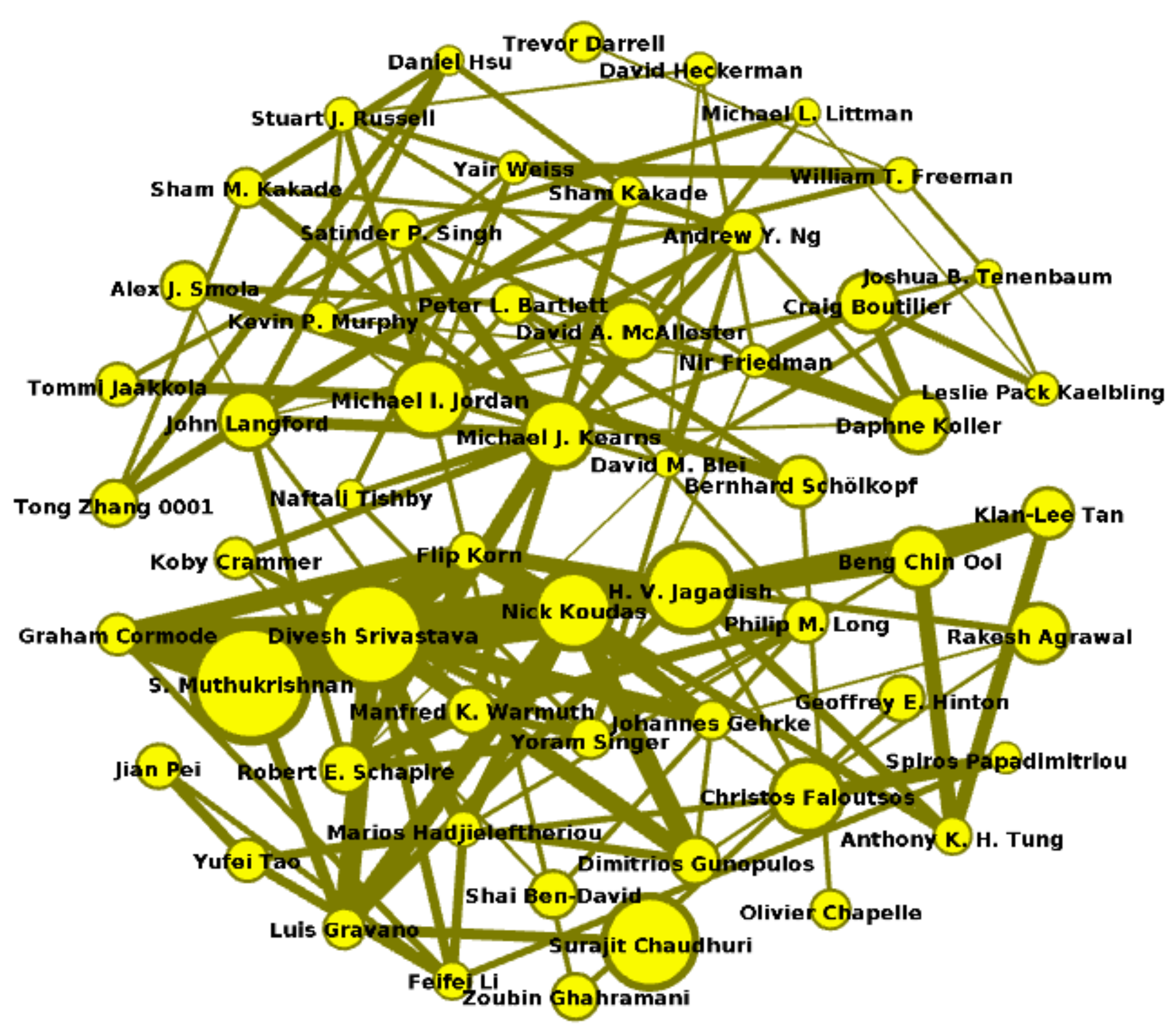}
\caption{Collaboration group on Databases, ML, AI, Pattern Recognition, Big Data.}
\label{fig:dbmlai}
\end{figure}

\section{Future Research}
In the near future, we intend to perform dynamic network analysis using the edge-closeness values computed at the end of the first phase. This would entail avoiding the recomputation of the community structure in a dynamic network and modify the edge-closeness values for some deletions and additions of edges.
Our algorithm can also be extended to the directed graphs and bipartite graphs by means of small modifications to the utility functions. However the algorithm is not yet run on the directed and bipartite benchmarks for its evaluation. This forms another scope for future work.

\bibliographystyle{plain}
\bibliography{ecai}

\begin{thebibliography}{10}

\bibitem{apt2014coordination}
Krzysztof~R Apt, Mona Rahn, Guido Sch{\"a}fer, and Sunil Simon, `Coordination
  games on graphs', in {\em Web and Internet Economics},  441--446, Springer,
  (2014).

\bibitem{bei2013trial}
Xiaohui Bei, Ning Chen, Liyu Dou, Xiangru Huang, and Ruixin Qiang, `Trial and
  error in influential social networks', in {\em Proceedings of the 19th ACM
  SIGKDD international conference on Knowledge discovery and data mining}, pp.
  1016--1024. ACM, (2013).

\bibitem{blondel_fast_2008}
Vincent~D Blondel, Jean-Loup Guillaume, Renaud Lambiotte, and Etienne Lefebvre,
  `Fast unfolding of communities in large networks', {\em Journal of
  Statistical Mechanics: Theory and Experiment}, {\bf 2008}(10),  P10008, (oct
  2008).

\bibitem{chen2010game}
Wei Chen, Zhenming Liu, Xiaorui Sun, and Yajun Wang, `A game-theoretic
  framework to identify overlapping communities in social networks', {\em Data
  Mining and Knowledge Discovery}, {\bf 21}(2),  224--240, (2010).

\bibitem{esquivel2011compression}
Alcides~Viamontes Esquivel and Martin Rosvall, `Compression of flow can reveal
  overlapping-module organization in networks', {\em Physical Review X}, {\bf
  1}(2),  021025, (2011).

\bibitem{gaur2008capacitated}
Daya~Ram Gaur, Krishnamurti. Ramesh, and Rajeev Kohli, `The capacitated max
  k-cut problem', {\em Mathematical Programming}, {\bf 115}(1),  65--72,
  (2008).

\bibitem{gopalan2013efficient}
Prem~K Gopalan and David~M Blei, `Efficient discovery of overlapping
  communities in massive networks', {\em Proceedings of the National Academy of
  Sciences}, {\bf 110}(36),  14534--14539, (2013).

\bibitem{granovetter1973}
Mark~S. Granovetter, `The strength of weak ties', {\em American Journal of
  Sociology}, {\bf 78}(6),  pp. 1360--1380, (1973).

\bibitem{gregory_finding_2010}
Steve Gregory, `Finding overlapping communities in networks by label
  propagation', {\em New Journal of Physics}, {\bf 12}(10),  103018, (oct
  2010).

\bibitem{guimera_functional_2005}
Roger Guimerà and Luís~A. Nunes~Amaral, `Functional cartography of complex
  metabolic networks', {\em Nature}, {\bf 433}(7028),  895--900, (feb 2005).

\bibitem{hoefer2007cost}
Martin Hoefer, {\em Cost sharing and clustering under distributed competition},
  Ph.D.\ dissertation, 2007.

\bibitem{kleinberg2006algorithm}
Jon Kleinberg and {\'E}va Tardos, {\em Algorithm design}, Pearson Education
  India, 2006.

\bibitem{lancichinetti2009benchmarks}
Andrea Lancichinetti and Santo Fortunato, `Benchmarks for testing community
  detection algorithms on directed and weighted graphs with overlapping
  communities', {\em Physical Review E}, {\bf 80}(1),  016118, (2009).

\bibitem{lancichinetti_benchmarks_2009}
Andrea Lancichinetti and Santo Fortunato, `Benchmarks for testing community
  detection algorithms on directed and weighted graphs with overlapping
  communities', {\em Physical Review E}, {\bf 80}(1), (jul 2009).

\bibitem{lancichinetti_detecting_2009}
Andrea Lancichinetti, Santo Fortunato, and János Kertész, `Detecting the
  overlapping and hierarchical community structure in complex networks', {\em
  New Journal of Physics}, {\bf 11}(3),  033015, (mar 2009).

\bibitem{lancichinetti_benchmark_2008}
Andrea Lancichinetti, Santo Fortunato, and Filippo Radicchi, `Benchmark graphs
  for testing community detection algorithms', {\em Physical Review E}, {\bf
  78}(4), (oct 2008).

\bibitem{lancichinetti_finding_2011}
Andrea Lancichinetti, Filippo Radicchi, José~J. Ramasco, and Santo Fortunato,
  `Finding statistically significant communities in networks', {\em {PLoS}
  {ONE}}, {\bf 6}(4),  e18961, (apr 2011).

\bibitem{monderer1996potential}
Dov Monderer and Lloyd~S Shapley, `Potential games', {\em Games and economic
  behavior}, {\bf 14}(1),  124--143, (1996).

\bibitem{narayanam2012game}
Ramasuri Narayanam and Yadati Narahari, `A game theory inspired, decentralized,
  local information based algorithm for community detection in social graphs',
  in {\em Pattern Recognition (ICPR), 2012 21st International Conference on},
  pp. 1072--1075. IEEE, (2012).

\bibitem{newman_modularity_2006}
M.~E.~J. Newman, `Modularity and community structure in networks', {\em
  Proceedings of the National Academy of Sciences}, {\bf 103}(23),  8577--8582,
  (jun 2006).

\bibitem{newman2001structure}
Mark~EJ Newman, `The structure of scientific collaboration networks', {\em
  Proceedings of the National Academy of Sciences}, {\bf 98}(2),  404--409,
  (2001).

\bibitem{palla_uncovering_2005}
Gergely Palla, Imre Derényi, Illés Farkas, and Tamás Vicsek, `Uncovering the
  overlapping community structure of complex networks in nature and society',
  {\em Nature}, {\bf 435}(7043),  814--818, (june 2005).

\bibitem{rosvall_information-theoretic_2007}
M.~Rosvall and C.~T. Bergstrom, `An information-theoretic framework for
  resolving community structure in complex networks', {\em Proceedings of the
  National Academy of Sciences}, {\bf 104}(18),  7327--7331, (may 2007).

\bibitem{rosvall_maps_2008}
M.~Rosvall and C.~T. Bergstrom, `Maps of random walks on complex networks
  reveal community structure', {\em Proceedings of the National Academy of
  Sciences}, {\bf 105}(4),  1118--1123, (jan 2008).

\bibitem{schelling1971dynamic}
Thomas~C Schelling, `Dynamic models of segregation?', {\em Journal of
  mathematical sociology}, {\bf 1}(2),  143--186, (1971).

\bibitem{weng2013virality}
Lilian Weng, Filippo Menczer, and Yong-Yeol Ahn, `Virality prediction and
  community structure in social networks', {\em Scientific reports}, {\bf 3},
  (2013).

\bibitem{xie_overlapping_2013}
Jierui Xie, Stephen Kelley, and Boleslaw~K. Szymanski, `Overlapping community
  detection in networks: The state-of-the-art and comparative study', {\em
  {ACM} Computing Surveys}, {\bf 45}(4),  1--35, (aug 2013).

\bibitem{xie_slpa:_2011}
Jierui Xie, Xiaoming Liu, and Boleslaw~K. Szymanski, `{SLPA}: Uncovering
  overlapping communities in social networks via a speaker-listener interaction
  dynamic process', {\em Proceedings of Data Mining technologies for
  computational collective intelligence workshop, at {ICDM}, Vancouver, {CA}},
  344--349, (dec 2011).

\bibitem{xie2011slpa}
Jierui Xie, Boleslaw~K Szymanski, and Xiaoming Liu, `Slpa: Uncovering
  overlapping communities in social networks via a speaker-listener interaction
  dynamic process', in {\em Data Mining Workshops (ICDMW), 2011 IEEE 11th
  International Conference on}, pp. 344--349. IEEE, (2011).

\bibitem{yanovskaya1968equilibrium}
EB~Yanovskaya, `Equilibrium points in polymatrix games', {\em Litovskii
  Matematicheskii Sbornik}, {\bf 8},  381--384, (1968).

\end{thebibliography}
\end{document}